\documentclass[journal,onecolumn,10pt]{IEEEtran}
	    	\usepackage{epstopdf}
       		 \usepackage{color}
       		 \usepackage{amsfonts,amsthm,amssymb,epsfig,graphicx,tabularx,multirow,fancyhdr}
       		 \usepackage{subfigure}
		\usepackage{relsize}
		\usepackage{multirow} 
		\usepackage[utf8]{inputenc}
		\usepackage[english]{babel}
		\newtheorem{theorem}{Theorem}[section]
		\newtheorem{lemma}[theorem]{Lemma}
		
		\newtheorem{corollary}[theorem]{Corollary}
		\theoremstyle{definition}
		\newtheorem{definition}{Def.}
		
		\theoremstyle{remark}
		
		\usepackage{cite}
		\usepackage{booktabs}
		
		\usepackage[cmex10]{amsmath}
		\usepackage{acronym}

		\usepackage{array}
		\usepackage{algpseudocode}

		\usepackage{algorithm}
		\usepackage{mdwmath}
		\usepackage{mdwtab}
		\usepackage{lipsum}
        \usepackage{mathtools}
        \usepackage{cuted}
		
		\ifCLASSOPTIONcompsoc
		 \usepackage[caption=false,font=normalsize,labelfont=sf,textfont=sf]{subfig}
		\else
		\usepackage[caption=false,font=footnotesize]{subfig}
		\fi
	
		\usepackage{setspace}
		\usepackage{amssymb}
		
\begin{document}

		\title{A Fast Data Driven Shrinkage of Singular Values for Arbitrary Rank Signal Matrix Denoising}
		\author{S.~K. Yadav, \IEEEmembership{Student Member,~IEEE,}
		        R.~Sinha,~\IEEEmembership{Member,~IEEE,}
		        and~P.~K.~Bora,~\IEEEmembership{Member,~IEEE}
		
		\IEEEcompsocitemizethanks{\IEEEcompsocthanksitem Authors are with the Department of Electronics and Electrical Engineering, Indian Institute of Technology Guwahati, Guwahati-781039, Assam, India.\protect\\
		E-mail: \{y.santosh,~rsinha,~prabin\}@iitg.ernet.in}}

		\markboth{}{Yadav \MakeLowercase{\textit{et al.}}: Data Driven Shrinkage of Singular Values}
	
	\IEEEtitleabstractindextext{%
	\begin{abstract}
	\boldmath
    Recovering a low-rank signal matrix from its noisy observation, commonly known as matrix denoising, is a fundamental inverse problem in statistical signal processing. Matrix denoising methods are generally based on shrinkage or thresholding of singular values with a predetermined shrinkage parameter or threshold. However, most of the existing adaptive shrinkage methods use multiple parameters to obtain a better flexibility in shrinkage. The optimal value of these parameters using  either cross-validation or Stein's principle. In both the cases, the iterative estimation of various parameters render the existing shrinkage methods computationally latent for most of the real-time applications. This paper presents an efficient data dependent shrinkage function whose parameters are estimated using Stein's principle but in a non-iterative manner, thereby providing a comparatively faster shrinkage method. In addition, the proposed estimator is found to be consistent with the recently proposed asymptotically optimal estimators using the results from random matrix theory. The experimental studies on artificially generated low-rank matrices and on the magnetic resonant imagery data, show the efficacy of the proposed denoising method.
	\end{abstract}
	
	\begin{IEEEkeywords}
	Singular value shrinkage, Linearly expandable threshold, SURE, Matrix Denoising.
	\end{IEEEkeywords}}
		
	\maketitle
	\IEEEdisplaynontitleabstractindextext
	\IEEEpeerreviewmaketitle

    \section{Introduction}
    \IEEEPARstart THE estimation of correlated signals from their noisy or incomplete observations has various applications in multivariate statistical signal processing~\cite{AndrewPattarson1976,klema1980singular,Scharf1991,Jolliffe:2002,MaBarziger2014}, machine learning~\cite{Vempala2007,MazumderHT10}, wireless communication~\cite{EdforsSBWB98,ChenWEH10}, etc., and has gained a lot of attention, recently. Specifically, it aims at recovering a low-rank signal matrix $\mathbf{X}\in\mathbb{R}^{n\times m}$ with a rank $r<<L=\min(n,m)$ from its noisy observation $\mathbf{Y}$ given as
    \begin{equation}
    \mathbf{Y}=\mathbf{X}+\mathbf{W},
    \label{eqn:data-model}
    \end{equation}
    where $\mathbf{W}$ is a random noise matrix. Following this data model, a number of methods have been proposed in literature for such tasks~\cite{Candes_Recht_2009,Cai_Candes_Shen2010,Zachariah2012,Chatterjee2012,Deledalle2012,Candes2013,Shabalin2013,Nadakuditi2014,
    Gavish2014,GavishDonoho2014,Josse2014,SVLT2015,ParekhSelesnic2016,Bigot2016}. In almost all of those methods, the observed data is first decorrelated using matrix factorization techniques, the signal and the noise components are identified, and then the noise components are removed to obtain an estimate of the underlying true signal matrix. Mostly, the singular value decomposition (SVD) is employed for matrix factorization. If the rank of underlying signal matrix is given, the well known Eckart-Young-Mirsky (EYM)~\cite{eckart1936approximation} theorem can be used to obtain the optimal $r$-rank approximation of $\mathbf{Y}$ as
    \begin{equation}
    \mathbf{\widehat{Y}}^{\text{EYM}}_{r}= \underset{\text{Rank}(\mathbf{Z})=r}{\operatorname{\arg\min}}\Big\|\mathbf{Y-Z}\Big\| := \sum_{i=1}^{r}y_i\tilde{\mathbf{u}}_i\tilde{\mathbf{v}}^{\smaller{T}}_i,
    \label{eqn:eym}
    \end{equation}
    where $\lVert\cdot\rVert$ is a unitary-invariant matrix norm. If $\mathbf{W}$ has independent and identically distributed (i.i.d.) Gaussian entries then $\mathbf{\widehat{Y}}^{\text{EYM}}_{r}$ is $r$-rank maximum likelihood estimate (MLE) of $\mathbf{X}$. Nevertheless, one can quest for the estimators, possibly the nonlinear ones, which incur lower mean-squared error (MSE) than MLE~\cite{stein1956,james1961}. The EYM uses no additional structures on matrix $\mathbf{X}$ other than low-rankness. The better estimation methods exploiting the additional structures such as non-negativity~\cite{Markovsky08,MarkovskyU13}, sparsity~\cite{JenattonOB10}, Hankel~\cite{LiLR97}, Toeplitz~\cite{FangY92a}, etc. also exist. In matrix denoising, the objective is to estimate the signal matrix such that $\text{MSE}(\mathbf{X},\widehat{\mathbf{X}}) =\mathbb{E}\|\mathbf{X}-\widehat{\mathbf{X}}\|_F^2$ is minimized. Surprisingly, without proper regularization, the EYM theorem as a procedure of finding the best low-rank representation tries to minimize $\text{MSE}(\mathbf{Y},\widehat{\mathbf{X}})$ and hence it can not be relied upon for faithful denoising. In practice, the rank $r$ is not known \emph{a priori} and usually estimated by using information theoretic criteria \cite{AIC,BIC,Wax,Hurvich} which further affect the accuracy of the EYM estimates. A natural alternative is to solve the regularized low-rank problem
    \begin{equation}
    \widehat{\mathbf{X}} = \underset{\mathbf{X}}{\operatorname{\min}}\|\mathbf{Y}-\mathbf{X}\|_F^2 + \lambda~\text{Rank}(\mathbf{X}),
    \label{eqn:rank-reg}
    \end{equation}
    where $\lVert\cdot\rVert_F$ denotes the Frobenius norm and $\lambda$ is a regularization parameter. For the non-convexity of the rank function, the close form solution of (\ref{eqn:rank-reg}) is not possible. As a heuristic solution, the singular value hard-thresholding (SVHT) estimator is often adopted. The SVHT is given as
    \begin{equation}
    \mathbf{\widehat{X}}^{\text{SVHT}}_{\mu}=\sum_{i=1}^{L}y_i\mathbf{1}
    _{(y_i>\mu)}\tilde{\mathbf{u}}_i\tilde{\mathbf{v}}^{\smaller{T}}_i,
    \label{eqn:svht}
    \end{equation}
    where $\mathbf{1}_{(y_i>\mu)}=1$ if $y_i>\mu$ and zero otherwise. The above optimization problem can be easily solved if the rank function is replaced by the nuclear norm, which is convex hull of the rank function. Such a relaxation guarantees the existence of global minima which can be obtained by sub-gradient methods. In ~\cite{CaiCS10}, it has been shown that the singular value soft-thresholding (SVST) estimator
    \begin{equation}
    \mathbf{\widehat{X}}^{\text{SVST}}_{\lambda}=\sum_{i=1}^{L}(y_i-\lambda)_{+}\tilde{\mathbf{u}}_i\tilde{\mathbf{v}}^{\smaller{T}}_i
    \label{eqn:svst},
    \end{equation}
    where $(\alpha)_{+} = \max(\alpha, 0)$ is the solution to the modified convex problem. Clearly, the SVST not only thresholds the singular values smaller than $\lambda$ but it also applies the shrinkage by this amount. The only difficulty that remains is in the selection of the optimal $\lambda$. It determines the extent of shrinkage applied and hence the trade-off between the bias and the variance. Higher values of $\lambda$ result in smaller variance but introduce larger bias, whereas lower values of $\lambda$ cause smaller bias but at the cost of larger variance. One may naturally choose a value of $\lambda$ that admits the minimum MSE. As $\mathbf{X}$ is unknown, the MSE is an unobservable quantity. Fortunately, due to Stein's lemma, in case of Gaussian noise, it is possible to construct an unbiased estimator of MSE that is free from the direct dependence on $\mathbf{X}$ with only mild assumptions on candidate estimators. Such estimators are commonly known as Stein's unbiased risk estimator (SURE). In~\cite{Candes2013}, the authors showed that the SVST follows these assumptions and proposed to determine its parameter using SURE. For involving only a single parameter, the SVST approach has very low complexity however suffers from two drawbacks. Firstly, for the shrinkage and thresholding being coupled, it is unable to provide optimal trade-off between the bias and the variance. Secondly, it shrinks all the dominant singular values equally despite the fact that each of them may not be equally important \cite{wnnm2014}. To deal with the first drawback, a logistic function based shrinkage that decouples the shrinkage and the thresholding by introducing additional parameters was proposed in \cite{SVLT2015}. Towards addressing the second drawback, in \cite{VerbanckJH15}, a non-linear estimator which applies less shrinkage on larger singular values compared to the smaller ones was proposed. Recently, in \cite{JosseS16}, the authors proposed an adaptive trace norm based estimator with two parameters. It deals with both the drawbacks of the SVST and provides more flexible shrinkage. The parameters of these estimators are often selected by minimizing SURE. On account of the inherent non-linearity, the methods employing SURE lack the closed form solution for the optimal value of the parameters. As a result, a grid search over fairly large search range is done for finding the optimal value of the parameters. Thus limiting the use of such estimators in practical applications. In this paper, we propose a linearly expandable singular value thresholding/shrinkage (SVLET) function which admits a closed form solution to corresponding SURE thereby circumventing the limitation of existing shrinkage functions outlined above and reducing the computational cost drastically.

    Apart from SURE-based estimation of parameters, the random matrix theory (RMT)~\cite{Meh2004} provides an alternative by which asymptotic mean squared error (AMSE) $\underset{n \rightarrow\infty} \lim\text{MSE}(\mathbf{X},\widehat{\mathbf{X}})$ can be inferred without the knowledge of matrix $\mathbf{X}$. Following this, the recent works in \cite{Shabalin2013,Nadakuditi2014,Gavish2014,GavishDonoho2014} tried to find the optimal shrinkage functions for matrix denoising but in asymptotic case, where the size of the matrices grows infinitely keeping the rank and the aspect ratio fixed. In \cite{Nadakuditi2014}, a data driven estimator which dominates the EYM estimator is proposed. In \cite{Gavish2014}, the authors showed that the optimal value of hard-threshold is a deterministic quantity equal to $(4/\sqrt{3})\sqrt{n}\sigma$ where $\sigma$ is the standard deviation of Gaussian entries in $\mathbf{W}$. In \cite{GavishDonoho2014}, the same authors proposed a general framework to derive the optimal shrinkage function for a number of loss functions. Due to the limitations of RMT, these methods happen to suppress all the singular values below a predefined threshold and are also limited to the asymptotic case. Departing from this large matrix low-rank setting, our focus is on estimating the signal matrix of finite size which is not necessarily a low-rank. Moreover, using the results from random matrix theory~\cite{Bigot2016}, we show that the proposed estimator maintains its asymptotic optimality in large matrix limits too.

    The rest of the paper is organized as follows. In Section~\ref{sec:MatrixEstimation}, the SVD based matrix denoising problem and its two basic solving methodologies are discussed in detail. The proposed shrinkage method is described in Section~\ref{sec:ProbMethod}. The performance comparison of the proposed and the existing shrinkage methods has been done in Section~\ref{sec:ExPerf}. Finally, the paper is concluded in Section~\ref{sec:Conclu}.

    \begin{table}
    \footnotesize
    \centering
    \caption{Singular Value Decompositions. Symbols bar and tilde over the decomposed entries indicate the noise-only matrix and the noisy (signal -plus-noise) matrix cases, respectively.}
    \label{tab:svd-table}
    \begin{tabular}{lll}
    \toprule
    $\mathbf{X} ~= \mathbf{U}\mathbf{\Sigma}_x\mathbf{V}^{\smaller{T}} =   \sum_{i=1}^{r}x_i\mathbf{u}_i\mathbf{v}^{\smaller{T}}_i$;& $\mathbf{U}^{\smaller{T}}\mathbf{U}=\mathbf{I}_n$, $\mathbf{V}\mathbf{V}^{\smaller{T}}=\mathbf{I}_n$\\
    \midrule
    $\mathbf{Y} ~= \mathbf{\widetilde{U}}\mathbf{\Sigma}_y\mathbf{\widetilde{V}}^{\smaller{T}} = \sum_{i=1}^{L}y_i\tilde{\mathbf{u}}_i\tilde{\mathbf{v}}^{\smaller{T}}_i$;& $\mathbf{\widetilde{U}}^{\smaller{T}}\mathbf{\widetilde{U}}=\mathbf{I}_n$, $\mathbf{\widetilde{V}}\mathbf{\widetilde{V}}^{\smaller{T}}=\mathbf{I}_n$  \\
        \midrule
        $\mathbf{W} = \mathbf{\bar{U}}\mathbf{\Sigma}_w\mathbf{\bar{V}}^{\smaller{T}} =  \sum_{i=1}^{L}w_i\bar{\mathbf{u}}_i\bar{\mathbf{v}}^{\smaller{T}}_i$;& $\mathbf{\bar{U}}^{\smaller{T}}\mathbf{\bar{U}}=\mathbf{I}_n$, $\mathbf{\bar{V}}\mathbf{\bar{V}}^{\smaller{T}}=\mathbf{I}_n$ \\
    \bottomrule
    \end{tabular}
    \end{table}

    \section{Matrix Estimation}
    \label{sec:MatrixEstimation}
    Using the notations for SVD as in Table~\ref{tab:svd-table}, a generalized shrinkage estimator for the data model (\ref{eqn:data-model}) can be given as
    \begin{equation}
    \mathbf{\widehat{X}}_{\eta}=\sum_{i=1}^{L}\eta(y_i)\tilde{\mathbf{u}}_i\tilde{\mathbf{v}}^{\smaller{T}}_i,
    	\label{eqn:gen-shrink}
    \end{equation}
    where $\eta:\mathbb{R_+}\mapsto\mathbb{R_+}$ is a point-wise shrinkage function. Our objective is to obtain an optimal  $\eta$ as
    \begin{equation}
    \eta^* = \underset{\eta}{\operatorname{\min}}\Big\|\sum_{i=1}^{r}x_i\mathbf{u}_i\mathbf{v}^{\smaller{T}}_i-
    \sum_{i=1}^{L}\eta(y_i)\tilde{\mathbf{u}}_i\tilde{\mathbf{v}}^{\smaller{T}}_i\Big\|_F^2.
    \label{eqn:mse-1}
    \end{equation}
    At first sight, it may appear that such a shrinkage function acts on noisy singular values only. However, a careful inspection reveals that the additive random noise also causes the angular deviation in the singular vectors.  In large matrix limits, these angular deviations are deterministic and ultimately depend on the singular values. Thus, the above formulation incorporates the effect of noise on both the singular values and the singular vectors. Moreover, the difficulty in (\ref{eqn:mse-1}) is that the first term is entirely unknown. Fortunately, there exist two ways to proceed further, i.e., we can either invoke RMT or SURE principles.
    \subsection{Random Matrix Theory}
    In order to recover a signal matrix from the observation corrupted by a random matrix using SVD, we need to understand the behavior of the singular values of random matrices and its effect on the singular values of the signal matrix when added with it. For pedagogical reasons, following lemmas are given.
    \begin{lemma}[Generalized Quarter-circle Distribution~\cite{Meh2004,anderson2010introduction}]
    \label{sv-dist} Consider an $n\times m_n$ random matrix $\mathbf{W}$ with elements $W_{ij}\stackrel{i.i.d.}{\sim}\mathcal{N}(0,\sigma^2)$ and the singular values $\{w_1\geq w_2\geq...\geq w_{n}\geq0\}$ where $m_n$ is an increasing sequence of integers dependent on $n$. Consider an asymptotic setting (large matrix limits) as $\underset{n \rightarrow\infty} \lim \frac{n}{m_n} \rightarrow\beta$ such that $0<\beta\leq1$ and $\sigma=1/\sqrt{m_n}$. Then, the empirical distribution of singular values of $\mathbf{W}$ converges to a non-random distribution with probability density function
    \begin{equation}
    f(w) = \frac{\sqrt{(w^2-\beta^2_-)(\beta^2_+-w^2)}}{\pi\beta w} \mathbf{1}_{[\beta_-,\beta_+]}(w),
    \end{equation}
    almost surely (a.s.), as $n\rightarrow\infty$ where $\beta_-=1-\sqrt{\beta}$, $\beta_+=1+\sqrt{\beta}$, and $\mathbf{1}_{[\cdot,\cdot]}$ is an indicator function. Moreover, $w_1\xrightarrow{a.s.}\beta_+$ and $w_{n}\xrightarrow{a.s.}\beta_-$ as $n\rightarrow\infty$.
    \end{lemma}
    The above lemma states that no singular value of $\mathbf{W}$ lies beyond the support [$\beta_-,\beta_+$]. The singular values in ($\beta_-,\beta_+$) are often referred to as \emph{bulk} with the two extremes, $\beta_-$ and $\beta_+$ being called as the \emph{bulk-edges}~\cite{Shabalin2013}. Clearly, any singular value beyond the bulk-edges indicates the presence of correlation in the data and so that of the signal\footnote{Here we have ignored an infinitesimal margin $\epsilon$ near the bulk edges.}.
    \begin{definition}[Bulk Shrinker, $\eta_{\beta_+}$] We call a shrinkage function $\eta(y)$ a \emph{bulk-shrinker} if $\eta(y) = 0$ for $y\leq \beta_+$.
    \end{definition}
    Clearly, if $\hat{\mathbf{x}}_{\beta_+}=[\eta_{\beta_+}(y_1),\eta_{\beta_+}(y_2), ..., \eta_{\beta_+}(y_n)]^{\smaller T}$ then one can take $r^*=\|\hat{\mathbf{x}}_{\beta_+}\|_0$ as an estimate of rank of $\mathbf{X}$, where $\|\cdot\|_0$ denotes the $l_0$-norm. By minimizing the AMSE as
    \begin{equation}
    \eta^*_{\beta_+} = \underset{\eta_{\beta_+}}{\operatorname{\min}}\underbrace{\underset{n \rightarrow\infty} \lim\Big\|\sum_{i=1}^{r}x_i\mathbf{u}_i\mathbf{v}^{\smaller{T}}_i-
    \sum_{i=1}^{r^*}\eta_{\beta_+}(y_i)\tilde{\mathbf{u}}_i\tilde{\mathbf{v}}^{\smaller{T}}_i\Big\|_F^2}_{\text{AMSE}_{\eta_{\beta_+}}(\mathbf{Y})}
    \label{eqn:amse}
    \end{equation}
    the authors in \cite{GavishDonoho2014} showed that a continuous bulk-shrinker admits the optimal solution
    \begin{equation}
    \eta^*_{\beta_+}(y_i) = x_i \langle u_i,\tilde{u}_i\rangle\langle v_i,\tilde{v}_i\rangle,~~~\forall i
    \label{eqn:optimal-shrink-1}
    \end{equation}
    where $\langle\cdot,\cdot\rangle$ denotes the inner product. We see that the optimal bulk-shrinker depends not only on the true singular values of signal matrix but also on the cosine of the angular deviation between corresponding singular vectors caused by noise matrix. Hence, it is important to study the effects of the noise matrix on the singular values as well as on the singular vectors. Lemma~\ref{sv-asyloc} presents the asymptotic locations of the singular values of $\mathbf{Y}$ relative to the singular values of $\mathbf{X}$ thereby giving a relation between $y_i$ and $x_i$. Lemma~\ref{sv-phase-tran} presents a relative phase transition in singular vectors of noise at bulk edges due to the presence of signal components giving the cosine of the angles required in (\ref{eqn:optimal-shrink-1}).
    \begin{lemma}[Asymptotic Location~\cite{benaych2012singular,Shabalin2013}]
    \label{sv-asyloc} Consider the data model in (\ref{eqn:data-model}) with Gaussian random matrix $\mathbf{W}$ of Lemma~\ref{sv-dist} and their SVD as defined in Table~\ref{tab:svd-table}. Assume that $y_1\geq y_2\geq...\geq y_{m_n}$ and $x_1\geq x_2\geq...\geq x_{r}\geq0$. For $1\leq i \leq r$, in large matrix limits, the asymptotic locations of singular values are related as
    \begin{equation}
    \underset{n \rightarrow\infty} \lim y_i \stackrel{a.s.}{=} \left\{
    		\begin{array}{ll}
    		 \rho(x_i), 	&  x_i>\beta^{1/4}\\
          	 \beta_+, 	&  0<x_i\leq\beta^{1/4}
          	\end{array}
    \right.
    \end{equation}
    where $r$ is the rank of $\mathbf{X}$ and
    \begin{equation}
    \rho(x) = \sqrt{\Big(x+\frac{1}{x}\Big)\Big(x+\frac{\beta}{x}\Big)}.
    \label{eqn:rho-y-relation}
    \end{equation}
    \end{lemma}
    \begin{lemma}[Phase-Transition~\cite{Paul2007Asymptotics,Shabalin2013}]
    \label{sv-phase-tran} Suppose $\mathbf{X}$ has distinct singular values $x_1> x_2>...> x_{r}>0$. Then, for $1\leq i, j \leq r$ and $x_i>\beta^{1/4}$
    \begin{equation}
    \underset{n \rightarrow\infty} \lim |\langle u_i,\tilde{u}_j\rangle| \stackrel{a.s.}{=} \left\{
    		\begin{array}{ll}
    		 \theta_u(x_i), 	& i=j\\
          	 0,		 		& i\neq j \\
          	\end{array}
    \right.
    \end{equation}
    and
    \begin{equation}
    \underset{n \rightarrow\infty} \lim |\langle v_i,\tilde{v}_j\rangle| \stackrel{a.s.}{=} \left\{
    		\begin{array}{ll}
    		 \theta_v(x_i), 	& i=j\\
          	 0,		 		& i\neq j \\
          	\end{array}
    \right.
    \end{equation}
    where
    \begin{equation}
    \theta_u(x) = \sqrt{\frac{x^4-\beta}{x^2(x^2+\beta)}}~\text{and}~\theta_v(x) = \sqrt{\frac{x^4-\beta}{x^2(x^2+1)}}.
    \end{equation}
    However, if $x_i\leq\beta^{1/4}$ or for $i,j\geq r$ we have $\underset{n \rightarrow\infty} \lim |\langle u_i,\tilde{u}_j\rangle| \stackrel{a.s.}{=} \underset{n \rightarrow\infty} \lim |\langle v_i,\tilde{v}_j\rangle|\stackrel{a.s.}{=}0$.
    \end{lemma}
    \begin{corollary}[Equivalence] \label{equivalence}For the data model (\ref{eqn:data-model}), as a direct consequence of Lemma~\ref{sv-asyloc}, the conditions $x_i>\beta^{1/4}$ and $y_i>\beta_+$ are equivalent.
    \end{corollary}
    Evidently, one can interchangeably use these conditions however, in the rest of the paper, we prefer the latter as $y_i$ is a known quantity. Using Lemma~\ref{sv-asyloc} and Lemma~\ref{sv-phase-tran}, (\ref{eqn:optimal-shrink-1}) can now be written as \begin{equation}
    \eta^*_{\beta_+}(y) = \left\{
    		\begin{array}{ll}
    		 \frac{\sqrt{(y_i^2-\beta-1)^2-4\beta}}{y_i}, & y_i>\beta_+~\text{and}~i\geq r^*\\
          	 0,		 		& \text{otherwise}.\\
          	\end{array}
    \right.
    \label{eqn:optimal-shrink-2}
    \end{equation}
    With this optimal bulk-shrinker, (\ref{eqn:gen-shrink}) results in an asymptotically optimal estimator denoted as $\mathbf{\widehat{X}}^*_{\beta_+}$ which achieves the lower bound on AMSE. From Fig.\ref{fig:shrinker-plots}(a), we see that it truncates the singular values smaller than $\beta_+$ and shrinks the singular values greater than $\beta_+$. For the optimal value of hard-threshold, $\mathbf{\widehat{X}}^{\text{SVHT}}_{\mu}$ and $\mathbf{\widehat{X}}^*_{\beta_+}$ are quite close to each other, hence a similar bias-variance trade-off is made in the case of asymptotic setting. However, in finite matrix limit or in relatively higher rank cases, this optimality is lost as the truncation of singular values may introduce a small bias.
    \subsection{SURE Principle}
    Consider a unitary invariant spectral function $F:\mathbb{R}^{n\times m}\mapsto\mathbb{R}^{n\times m}$ such that $F_{\eta}(\mathbf{Y})=\mathbf{\widetilde{U}}F_{\eta}(\mathbf{\Sigma})\mathbf{\widetilde{V}}^{\smaller{T}}$ where $F_{\eta}(\mathbf{\Sigma})=\text{diag}(\eta(y_1),\eta(y_2), ..., \eta(y_m))$. In \cite{Candes2013}, the authors showed that for a weakly differentiable spectral function, an unbiased estimator of the MSE for the data model in (\ref{eqn:data-model}) with $W_{i,j}\stackrel{iid}{\sim}\mathcal{N}(0,\sigma^2)$, can be given as
    \begin{equation}
    \text{SURE}_{\eta}{(\mathbf{Y})} = -nm\sigma^2 +\|\mathbf{Y}-F_{\eta}(\mathbf{Y})\|_{F}^{2} +2\sigma^2 \mbox{div} (F_{\eta}(\mathbf{Y})),
    \label{eqn:gen_sure}
    \end{equation}
    where div($F_{\eta}(\mathbf{Y})$) is the divergence of spectral estimator as given in the following theorem.
    \begin{theorem}[Divergence of Spectral Function~\cite{Candes2013}] \label{theorem-div} Suppose $\mathcal{F}$ is a set of simple matrices with full rank such that Lebesgue measure of its complementary set is zero. Consider a unitary-invariant spectral function $F:\mathcal{F} \mapsto \mathbb{R}^{n\times m}$ of the form
    \begin{equation}
    F_{\eta}(\mathbf{Y}) =\mathbf{\widetilde{U}}F_{\eta}(\mathbf{\Sigma})\mathbf{\widetilde{V}}^{\smaller{T}}:=\sum_{i=1}^{L}\eta(y_i)\tilde{\mathbf{u}}_i\tilde{\mathbf{v}}_i.
    \label{eqn:spectral-fun}
    \end{equation}
    If $F_{\eta}$ is differentiable, the divergence of $F_{\eta}(\mathbf{Y})$ can be given as
    \begin{align}
    \emph{div}(F_{\eta}(\mathbf{Y}))=\sum_{i=1}^{L}\eta'(y_i)+|n-m|\sum_{i=1}^{L}\frac{\eta(y_i)}{y_i}+2\mathop{\sum_{i=1}\sum_{j=1,}}_{i\neq j}^{L}\frac{y_i~\eta(y_i)}{y_i^2-y_j^2},
    \label{eqn:div}
    \end{align}
    \end{theorem}
    where $\eta' = \partial\eta/\partial y_i$. The differentiability assumption in the above can be relaxed, without any loss, to the differentiability in a spectral neighborhood which depends on differentiability of all factors of SVD of $\mathbf{Y}$ \cite{magnus1985differentiating,Sun2002NonsmoothMV,Lewis2001TwiceDS}. Hence, we only require $\mathbf{Y}$ to be simple (no repeated singular values) and full-rank. For the data model (\ref{eqn:data-model}), due to the random Gaussian noise $\mathbf{Y}$ is full-rank with high probability for any signal matrix $\mathbf{X}$.

    Note that SURE does not depend on $\mathbf{X}$ and hence it can be used to directly find the estimate of singular values. However, minimizing SURE for $n$ number of singular values is NP hard. This is why a parametric form of shrinkage function, with much lower number of free parameters than $n$, is specified and then the optimal parameters are found by minimizing SURE. For instance, the SVST $\eta^{\text{svst}}_{\lambda}(y_i) = (y_i-\lambda)_{+}$ is one of the examples of such an estimator with only one parameter $\lambda$ giving rise to the estimator in (\ref{eqn:svst}), where the optimal value of $\lambda$ is determined as
    \begin{equation}
    \lambda_{opt}=\underset{\lambda}{\operatorname{\min}}~\text{SURE}_{\eta}(\mathbf{Y}).
    \label{eqn:lambda-tune}
    \end{equation}
    Other approaches include the adaptive trace norm (ATN)~\cite{JosseS16} estimator $\mathbf{\widehat{X}}^{\text{ATN}}_{(\tau,\gamma)}$ with the shrinkage function
    \begin{equation}
    \eta^{\text{atn}}(\tau,\gamma)= y_i \Big(1-\frac{\tau^\gamma}{{y_i}^\gamma}\Big)_{+},
    \label{eqn:atn}
    \end{equation}
    and the SVLT~\cite{SVLT2015} estimator $\mathbf{\widehat{X}}^{\text{SVLT}}_{(p_1,p_2,p_3)}$ with the shrinkage function
    \begin{equation}
    \eta^{\text{svlt}}(p_1,p_2,p_3)= \Big(\frac{y_i}{1+e^{p_1(i-p_2)}}-p_3\Big)_{+},
    \label{eqn:svlt}
    \end{equation}
    having two and three parameters, respectively. The shape of these shrinkage estimators can be seen in Figs.~\ref{fig:shrinker-plots}(b) and \ref{fig:shrinker-plots}(c). For $\gamma=1$, the ATN reduces to the SVST and for $\gamma=\infty$ it becomes the SVHT where $\tau$ works as their threshold parameter. Similarly, for smaller value of $p_1$, the SVLT works as the SVST whereas for a larger value, it works as the SVHT and the parameters $p_2$ and $p_3$ determine the threshold and the extent of shrinkage, respectively. Hence, both of these estimators can be seen as a generalization of the SVHT and the SVST. Evidently, as we increase the number of parameters, more flexibility and control over shrinkage can be achieved. Since, each parameter is determined by minimizing corresponding SURE one-by-one, as the number of free parameters increases SURE becomes too complex to fit into real-time applications. This is because SUREs corresponding to these estimators do not impart a closed form solution for their parameters leaving us to rely upon a heuristic search for their optimal values in a given range. In practical applications, the use of such estimator leads to high computational cost (refer to \cite{Candes2013} for latency in finding optimal $\lambda$ in magnetic resonant imaging (MRI) denoising application). Hence, it is desirable to develop an estimator circumventing this limitation of SURE based parameter estimation. With this motivation, we propose a novel shrinkage estimator for which SURE imparts a closed form solution for its parameters and performs better than the above estimators in finite limits without loosing its asymptotic optimality in large matrix limits.

    \section{Proposed Method}
    \label{sec:ProbMethod}
    \begin{figure}[t]
    \includegraphics[width=\linewidth]{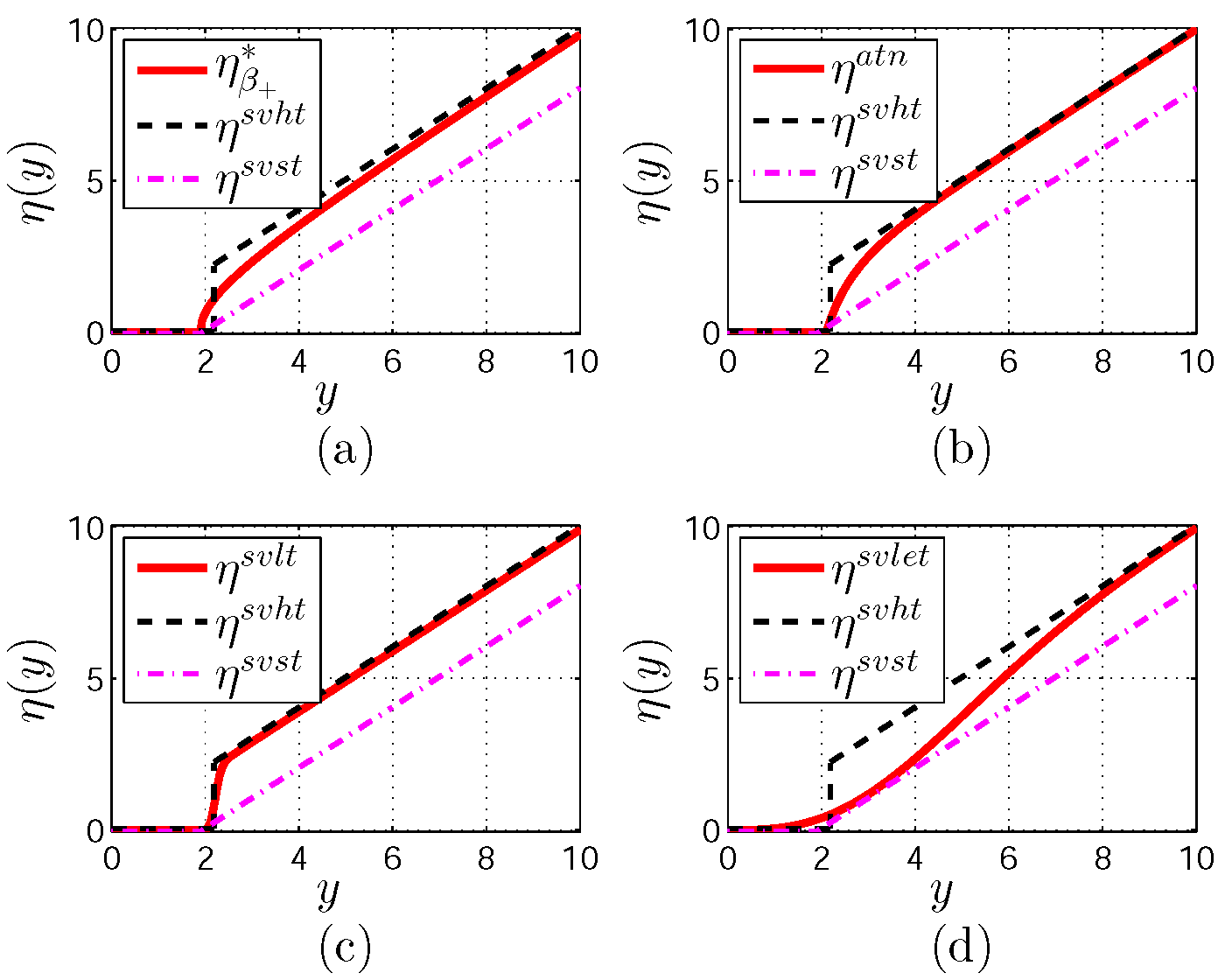}
    \caption{Comparison of existing singular value shrinkage functions. A typical example of (a) asymptotically optimal bulk-shrinkage function~\cite{GavishDonoho2014}, (b) ATN~\cite{JosseS16} ($\tau=2.2,\gamma = 5$), (c) SVLT~\cite{SVLT2015} ($p_1=0.8, p_2=2.2, p_3=0.1$), and (d) proposed SVLET shrinkage functions ($\mathbf{a}=[1,-1]^{\smaller T}, K=2, T=3$) plotted with SVHT ($\mu=2.18$) and SVST ($\lambda=1.98$). The values in the braces are the parameters of corresponding shrinkage function.}
    \label{fig:shrinker-plots}
    \end{figure}	
    \subsection{Proposed Singular Value Shrinkage Function} What hinders in obtaining closed form optima of SURE is the necessary nonlinearity in the shrinkage functions employed in existing estimators. This motivated us to linearly parameterize the nonlinear shrinkage function as
    \begin{equation}
    \eta(y) = \sum_{k=1}^{K} a_k\phi_k(y) = \mathbf{\Phi}^{\mathsmaller T}\mathbf{a}
    \label{eqn:PropShrink}
    \end{equation}
    where $\mathbf{a}= [a_1, a_2, ..., a_K]^{\mathsmaller T}$ is the shrinkage parameter vector, $\mathbf{\Phi} = [\phi_{1}(y), \phi_{2}(y), ...,\phi_K(y)]^{\mathsmaller T}$ is hitherto unspecified  bases, and $K$ is the order of linearization. Plugging this in (\ref{eqn:gen_sure}) and using (\ref{eqn:div}), the optimal value of parameter  $\mathbf{a}$ of the proposed shrinkage function can be obtained by minimizing SURE in closed form as given below. Using the gradient method to find the optimal value of $a_k$, we can write
    \begin{align}
    0& = \frac{1}{2}\frac{\partial}{\partial a_k}~\text{SURE}_{\eta}(\mathbf{Y})  \nonumber \\
    & = \underbrace{\frac{1}{2}\frac{\partial}{\partial a_k}~\|\mathbf{Y}-F_{\eta}(\mathbf{Y})\|_{F}^{2}}_{\text{R}_{k,1}}~+~\sigma^2 \underbrace{\frac{\partial}{{\partial a_k}}~\mbox{div} (F_{\eta}(\mathbf{Y}))}_{\text{R}_{k,2}}.
    \label{eqn:R1plusR2}
    \end{align}
    Solving for $\text{R}_{k,1}$ and $\text{R}_{k,2}$ as,
    \begin{align}
    \text{R}_{k,1}& = \sum_{i=1}^{L}\frac{1}{2}\frac{\partial}{\partial a_k} \Big(y_i-\sum_{k=1}^{K} a_k\phi_k(y)\Big)^2 \tag{Using (\ref{eqn:PropShrink})}~ \nonumber \\
    	           & = \sum_{l=1}^{K}\sum_{i=1}^{L} \phi_k(y_i)\phi_l(y_i)a_l-\sum_{i=1}^{L}y_i\phi_k(y_i),\nonumber
    \end{align}
    \begin{align}
    \text{R}_{k,2}& = \sum_{i=1}^{L}\Big(\phi_k'(y_i) + |n-m|\frac{\phi_k(y_i)}{y_i}+2\mathop{\sum_{j=1,}^{L}}_{j\neq i}\frac{y_i \phi_k(y_i)}{y_i^2-y_j^2}\Big) \nonumber
    \end{align}
    and collecting the terms in (\ref{eqn:R1plusR2}), we have
    \begin{align}
    	 0& = \sum_{l=1}^{K}\underbrace{\sum_{i=1}^{L}\phi_k(y_i)\phi_l(y_i)}_{M_{k,l}}~a_l-c_k,
    \label{eqn:collect-in-it}
    \end{align}
    where
    \begin{align}
    c_k = \sum_{i=1}^{L}\Big(y_i-\frac{|n-m|\sigma^2}{y_i}- \mathop{\sum_{j=1,j\neq i}^{L}}\frac{2~\sigma^2y_i}{y_i^2-y_j^2}\Big)\phi_k(y_i)-\sum_{i=1}^{L}\sigma^2\phi_k'(y_i)
    \label{eqn:equation_of_cK}
    \end{align}
    is free from any unknown parameter. Repeating the above differentiation for all $k\in[1,K]$ and collecting the corresponding equations, as in (\ref{eqn:collect-in-it}), in matrix form we have
    \begin{equation}
    0 = \mathbf{M}\mathbf{a}-\mathbf{c} \Longrightarrow \mathbf{a}_{opt} = \mathbf{M}^{-1}\mathbf{c},
    \label{eqn:optimal_a}
    \end{equation}
    where
    \[\mathbf{M} =
    \begin{bmatrix}
    \sum_{i=1}^{L}\phi_{1}(y_i)\phi_{1}(y_i) & \cdots & \sum_{i=1}^{L}\phi_{1}(y_i)\phi_{K}(y_i)\\
    	\vdots & \ddots & \vdots \\
    	 \sum_{i=1}^{L}\phi_{K}(y_i)\phi_{1}(y_i) & \cdots & \sum_{i=1}^{L}\phi_{K}(y_i)\phi_{K}(y_i)
    \end{bmatrix}\]
    and  $\mathbf{c}= [c_1, c_2, ..., c_K]^{\mathsmaller T}$. Thus, for a given $K$ and $\phi_k(\cdot)$, finding the value of shrinkage parameter vector $\mathbf{a}$ has reduced to solving a system of linear equations. Note that the order of linearization ($K$) can be arbitrarily increased, as long as matrix $\mathbf{M}$ is invertible, providing sufficient flexibility on the shape of shrinkage function. What remains still unspecified is the choice of basis function $\phi_{k}(\cdot)$. Authors in \cite{LuisierBluUnser2007} applied similar shrinkage function on multiscale DWT coefficients and named it as SURELET. On experimenting over several bases, they found the derivative of Gaussian (DOG) to be the best choice. In this work, we too adopt the DOG as the basis. The functional form of DOG is
    \begin{equation}
    \phi_k(y) = y e^{-(k-1)\frac{y^2}{2T^2}}
    \end{equation}
     where $T$ is a constant that depends on the strength of noise, i.e., $T=C\sigma$ with $C>0$ being a constant. The value of constant $C$ determines a smooth transition between noise dominated and signal dominated singular values.  It depends on the dynamic range of singular values of underlying noise free signal matrix and hence can be experimentally fixed for a data set. In Fig.\ref{fig:shrinker-plots}(d), a typical example of the proposed shrinkage function with DOG basis is compared with few existing shrinkage functions. Note that, unlike existing shrinkage functions, the proposed shrinkage function does not abruptly truncate the singular values, thus providing a better trade-off between the bias and the variance. For larger singular values, it is able to approximate the SVHT well, hence introduces a lower bias while keeping all the advantages of the SVST intact in case of lower singular values. As it does not necessarily truncate the smaller singular values to achieve the bias-variance trade-off, it can denoise matrices of any rank. In the rest of the paper, we term the proposed estimator with such linearization of shrinkage function as SVLET to differentiate it from the SURELET. Further, we denote the optimal SVLET shrinkage function and the corresponding shrinkage estimator as $\eta^{\text{svlet}}(y)$ and $\mathbf{\widehat{X}}_{\mathbf{a},T,K}^{\text{SVLET}}$, respectively.

    \subsection{Asymptotic Behavior of SVLET}
    In this subsection, we analyze the behavior of SVLET in asymptotic setting. In the following, we show that the proposed SVLET boils down to the asymptotically optimal estimator (\ref{eqn:optimal-shrink-2}) proposed in \cite{Shabalin2013,GavishDonoho2014}. For that, we apply the asymptotic conditions ($\sigma=1/\sqrt{m}$, and $n/m\rightarrow\beta$, with $0<\beta\leq1$ as $n\rightarrow\infty$) on SURE in (\ref{eqn:gen_sure}). Using Theorem~\ref{theorem-div}, an asymptotic SURE generated by bulk-shrinker $\eta_{\beta_+}$ can be given as
   \begin{align}
    \text{ASURE}_{\eta_{\beta_+}}(\mathbf{Y})&=\underset{n\rightarrow\infty}\lim \text{SURE}_{\eta_{\beta_+}}(\mathbf{Y}) \nonumber\\
     & =\underset{n\rightarrow\infty} \lim\Big\{-n + \sum_{i=1}^{r^*}\Big(y_i-\eta_{\beta_+}(y_i)\Big)^2+\frac{2}{m}~\mbox{div}(F_{\eta_{\beta_+}}(\mathbf{Y}))\Big\}, \nonumber\\
     & = \underset{n\rightarrow\infty}\lim (-n) + \sum_{i=1}^{r^*} \Big(y_i-\eta_{\beta_+}(y_i)\Big)^2+\sum_{i=1}^{r^*}\Big(2(1-\beta)\frac{\eta_{\beta_+}(y_i)}{y_i} +4\beta\underset{n\rightarrow\infty}\lim\frac{1}{n}\sum_{j=1,i\neq j}^{r^*} \frac{y_i\eta_{\beta_+}(y_i)}{y_i^2-y_j^2}\Big).
     \label{eqn:asym-sure-1}
    \end{align}
    Note that, in the above we used $\underset{n\rightarrow\infty}\lim\frac{\eta'_{\beta_+}(y_i)}{m} \stackrel{a.s.}{=} 0$ as $m=\beta^{-1}n$. Further, we need the following lemma.
    \begin{lemma}\cite{Bigot2016} \label{asym-div} In asymptotic setting of Lemma~\ref{sv-asyloc}, for $1\leq i \leq r^*$ and $y_i>\beta_+$, we have, almost surely,
    \begin{equation}
    \underset{n\rightarrow\infty} \lim \frac{1}{n}\sum_{j=1;j\neq i}^{n}\frac{y_i}{y^2_i-y^2_j} = \frac{1}{\rho(x_i)}\Big(1+\frac{1}{x^2_i}\Big).
    \end{equation}
    \end{lemma}
    Lemma~\ref{asym-div} provides the recipe for the following theorem which guarantees that, in asymptotic setting, the SVLET estimator generated by the bulk-shrinker is optimal in two ways; firstly, it achieves the true mean square error and secondly, the resulting estimator turns out to be exactly same as the optimal bulk-shrinker given in (\ref{eqn:optimal-shrink-2}).
    \begin{theorem} \label{theorem-asymp_optim} In asymptotic setting, the SVLET estimator generated by the bulk-shrinker achieves the true mean-squared error estimator and results in an optimal shrinkage function. Specifically,
    \begin{align}
    \underset{\eta_{\beta_+}}{\operatorname{\min}}~\text{ASURE}_{\eta_{\beta_+}}(\mathbf{Y}) = \underset{\eta_{\beta_+}}{\operatorname{\min}}~\text{AMSE}_{\eta_{\beta_+}}(\mathbf{Y}). \label{eqn:part1}\\
    \mbox{Equivalently,}  ~~~~~~~  \underset{n\rightarrow\infty} \lim \eta_{\beta_+}(\mathbf{y}) = \mathbf{\eta}^*_{\beta_+}(\mathbf{y}) \label{eqn:part2}
    \end{align}
    \end{theorem}
    \begin{proof} We first solve the left-hand-side of (\ref{eqn:part1}) for the minima of ASURE. At optima, we have
    \begin{equation}
     \begin{aligned}
     0& = \frac{1}{2}\frac{\partial}{\partial a^{\infty}_{k,opt}}\text{ASURE}_{\eta_{\beta_+}}(\mathbf{Y}).\nonumber \\
     \end{aligned}
     \end{equation}
     From (\ref{eqn:asym-sure-1}), we can write
     \begin{equation}
     \begin{aligned}
     &0 = \sum_{i=1}^{r^*}\Big\{\frac{1}{2}\frac{\partial}{\partial a^{\infty}_{k,opt}}\Big(y_i-\eta_{\beta_+}(y_i)\Big)^2 + \frac{(1-\beta)}{y_i}\frac{\partial \eta_{\beta_+}(y_i)}{\partial a^{\infty}_{k,opt}} + 2\beta \Big(\underset{n\rightarrow\infty} \lim \frac{1}{n}\sum_{j=1,i\neq j}^{r^*} \frac{y_i}{y_i^2-y_j^2}\Big)\frac{\partial \eta_{\beta_+}(y_i)}{\partial a^{\infty}_{k,opt}}\Big\}. \nonumber
     \end{aligned}
     \end{equation}
     Using Lemma~\ref{asym-div}, we have
     \begin{equation}
     \begin{aligned}
     0& = \sum_{i=1}^{r^*}\eta_{\beta_+}(y_i)\phi_k(y_i) + \sum_{i=1}^{r^*}\Big\{-y_i + \frac{(1-\beta)}{y_i} + \frac{2\beta}{\rho(x_i)}\Big(1+\frac{1}{x^2_i}\Big)\Big\}\phi_k(y_i) \nonumber\\
     \end{aligned}
     \end{equation}
     From Lemma~\ref{sv-asyloc}, for $y_i>\beta_+$, $y_i = \rho(x_i)$. Hence,
     \begin{equation}
     \begin{aligned}
     &\sum_{l=1}^{K}\underbrace{\sum_{i=1}^{r^*}\phi_k(y_i)\phi_l(y_i)}_{M^{\infty}_{k,l}}a^{\infty}_{l,opt}=\sum_{i=1}^{r^*}\frac{x_i^2 (y_i^2-\beta-1) -2 \beta}{x_i^2 y_i}\phi_k(y_i).\nonumber\\
     \end{aligned}
     \end{equation}
     Let $b_i=y_i^2-\beta-1$, then from (\ref{eqn:rho-y-relation}), $x_i^2 = \frac{b_i\pm \sqrt{b_i^2-4\beta}}{2}$
     \begin{equation}
     \begin{aligned} \sum_{l=1}^{K}M^{\infty}_{k,l}~a^{\infty}_{l,opt}=\sum_{i=1}^{r^*}\frac{1}{y_i}\sqrt{b_i^2-4\beta}~\phi_k(y_i).\nonumber\\
     \end{aligned}
     \end{equation}
     From (\ref{eqn:optimal-shrink-2}),
     \begin{equation}
     \begin{aligned}
     \sum_{l=1}^{K}M^{\infty}_{k,l}a^{\infty}_{l,opt}=\underbrace{\sum_{i=1}^{r^*}\eta^*_{\beta_+}(y_i)\phi_k(y_i)}_{c^{\infty}_{k}}.
     \end{aligned}
     \end{equation}
    Now, differentiating ASURE w.r.t. $a^{\infty}_{k,opt}$ for all $k$ and collecting the terms in matrix form yields
    \begin{equation}
    \mathbf{M}^{\infty}\mathbf{a}^{\infty}_{opt} = \mathbf{c}^{\infty}~~ \Longrightarrow \mathbf{a}^{\infty}_{opt} = (\mathbf{M}^{\infty})^{-1}\mathbf{c}^{\infty},
    \label{eqn:asure-minima}
    \end{equation}
    where $\mathbf{c}^{\infty} = (\mathbf{\Phi}^{\infty})^{\smaller T}\mathbf{\eta}^*_{\beta_+}(\mathbf{y})$. On solving (\ref{eqn:part1}) from the right-hand-side, we have
    \begin{equation}
    \begin{aligned}
    0& = \frac{1}{2}\frac{\partial}{\partial a^{\infty}_{k,opt}}\text{AMSE}_{\eta_{\beta_+}}(\mathbf{Y})\nonumber \\
    & = \sum_{i=1}^{r^*}\frac{1}{2}\frac{\partial}{\partial a^{\infty}_{k,opt}}\Big(x_i^2+\eta^2_{\beta_+}(y_i)\Big)^2-\underset{n\rightarrow\infty} \lim\sum_{i=1}^{r^*}\sum_{j=1}^{r^*}x_i\langle u_i,\tilde{u}_j\rangle\langle v_i,\tilde{v}_j\rangle\frac{\partial \eta_{\beta_+}(y_i)}{\partial a^{\infty}_{k,opt}}\nonumber\\
    \end{aligned}
    \end{equation}
    Using (\ref{sv-phase-tran}), we can write
    \begin{equation}
    \begin{aligned}
    0& = \sum_{l=1}^{K}\sum_{i=1}^{r^*}\phi_k(y_i)\phi_l(y_i)a^{\infty}_{l,opt}-\sum_{i=1}^{r^*}x_i\theta_u(x_i)\theta_v(x_i)\phi_k(y_i)\nonumber\\
    \end{aligned}
    \end{equation}
    Using (\ref{eqn:optimal-shrink-2}), we can write
    \begin{equation}
    \begin{aligned}
    0& = \sum_{l=1}^{K}\sum_{i=1}^{r^*}\phi_k(y_i)\phi_l(y_i)a^{\infty}_{l,opt}-\sum_{i=1}^{r^*}\eta^*_{\beta_+}(y_i)\phi_k(y_i)\nonumber\\
    \end{aligned}
    \end{equation}
    \begin{equation}
    \begin{aligned}
    &\sum_{l=1}^{K}\sum_{i=1}^{r^*}\phi_k(y_i)\phi_l(y_i)a^{\infty}_{l,opt}=\sum_{i=1}^{r^*}\eta^*_{\beta_+}(y_i)\phi_k(y_i)
    \end{aligned}
    \end{equation}
    Now, differentiating AMSE w.r.t.  $a^{\infty}_{k,opt}$ for all $k$ and collecting the terms in matrix form again gives (\ref{eqn:asure-minima}) hence (\ref{eqn:part1}) is proved. With this result, to prove (\ref{eqn:part2}), let us find the optimal bulk-shrinker for the SVLET estimator as
    \begin{align}
    \underset{n\rightarrow\infty} \lim \eta_{\beta_+}(\mathbf{y}) & = \sum_{i=1}^{r^*}\phi^{\infty}_k(y_i)a^{\infty}_{k, opt}\nonumber\\
    & = \mathbf{\Phi}^{\infty}\mathbf{a}^{\infty}_{opt}\nonumber\\
    & = \mathbf{\Phi}^{\infty}(\mathbf{M}^{\infty})^{-1}\mathbf{c}^{\infty}\nonumber\\
    & = \mathbf{\Phi}^{\infty}(\mathbf{M}^{\infty})^{-1}\mathbf{\Phi}^{\infty\smaller T}\mathbf{\eta}^*_{\beta_+}\nonumber\\
    & = \mathbf{\eta}^*_{\beta_+} \tag{Since, $\mathbf{\Phi}^{\infty}(\mathbf{M}^{\infty})^{-1}\mathbf{\Phi}^{\infty\smaller T}=I_n$}\nonumber
    \end{align}
    This completes the proof.
    \end{proof}
    \section{Results and Discussion}
    \label{sec:ExPerf}
    In this section, we empirically evaluate the performance of the proposed SVLET-based matrix estimation method through simulations and compare it with some of the recently proposed shrinkage methods. The empirical evaluation is primarily done on artificially generated matrices of different ranks. For that, an $r$-rank signal matrix is generated as $\mathbf{X}=\mathbf{L}\mathbf{R}^{\smaller T}$ where $\mathbf{L}$ and $\mathbf{R}$ are, respectively, $n\times r$ and $m\times r$ random matrices whose entries are independently drawn from zero mean unit variance Gaussian distribution. To obtain the corresponding observation matrix $\mathbf{Y}$, we add zero mean white Gaussian noise (AWGN) of known variances to each element of $\mathbf{X}$. Denoising is then performed on the observation matrix and the performance is measured in terms of normalized MSE
    \begin{equation}
    \text{NMSE} = \frac{1}{P}\sum_{p=1}^{P}\frac{\|\mathbf{\widehat{X}}^{(p)}_{\eta}-\mathbf{X}^{(p)}\|^2_F}{\|\mathbf{X}^{(p)}\|^2_F},
    \end{equation}
    where $P$ is the number of realizations. In this section, all the performances reported are averaged over $10$ noise realizations, i.e, $P=10$. Later, in this section, we test the efficacy of the SVLET in some real applications, also.

    \subsection{Sensitivity of the parameters}
    Before comparing the SVLET with a existing methods, we examine the effects of its two fixed parameters $C$ and $K$ on denoising performance. For that, we generate the test data of size $50\times 50$ with $1\leq r\leq 50$ as discussed above and noised at various signal-to-noise ratio (SNR) ranging from $0.5-4.0$. The SVLET with different values of $C$ and $K$ are employed for denoising. In Fig.~\ref{fig:CKvsNMSE}(a), the normalized MSE, averaged over all ranks and $1\leq K \leq 5$, is plotted against different values of $C$. In Fig.~\ref{fig:CKvsNMSE}(b), the normalized MSE, averaged over all ranks and $1\leq C \leq 20$, is plotted for $1\leq K \leq 5$. We see that increasing the value of $K$ beyond $2$ doesn't help much. The rationale behind this behavior is the sufficient adaptation of shrinkage function in capturing the singular value profile of original signal matrix by the parameter $\mathbf{a}$, only. Hence, the choice $C=10$ and $K=2$ are fair enough to be fixed in the rest of the experiments with this data set.
    \begin{figure}
    \includegraphics[width=\linewidth]{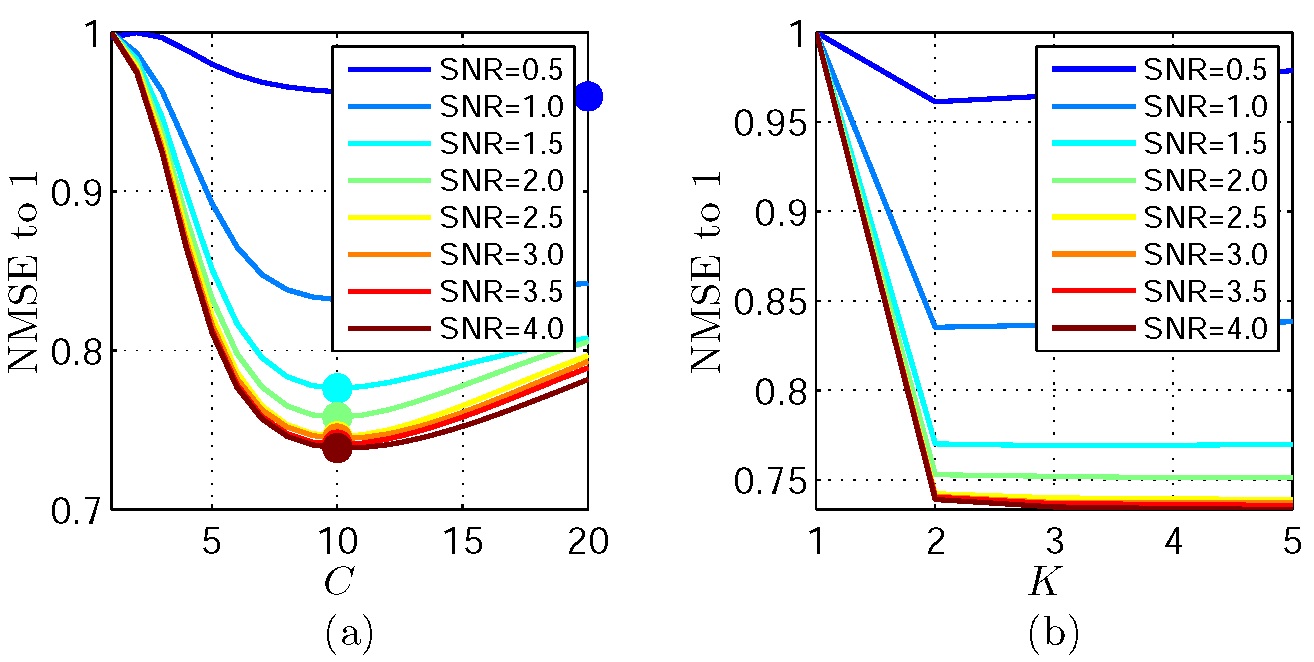}
    \caption{Sensitivity of the parameters of the proposed SVLET shrinkage estimator on the denoising performance. Effect on the NMSE of varying (a) $C$ and (b) $K$. In each case, the averaged NMSE computed over all possible ranks of the underlying clean matrix is reported.}
    \label{fig:CKvsNMSE}
    \end{figure}

    \subsection{Performance Evaluation on Matrices}
    In the following, the denoising performance of SVLET is compared with the state-of-the-art methods using the artificially generated data set as discussed above.
    \subsubsection{Comparison with Asymptotic Methods} The denoising performance of the SVLET is compared with that of the asymptotic estimators proposed in \cite{GavishDonoho2014,Gavish2014} for SNR = $0.5$, $1.0$, $2.0$, and $4.0$. As the RMT based methods are calibrated for standard deviation $\sigma = 1/\sqrt{n}$ of noise in asymptotic setting, we adjust the scale of given observed matrix before applying these estimator. Thus, the final estimate is obtained as $\sqrt{n}\sigma\mathbf{\widehat{X}}^*_{\beta_+}(\mathbf{Y}/(\sqrt{n}\sigma))$. For the SVHT, the optimal value of hard-threshold is $\mu^*=4/\sqrt{3}$ and for the SVST, we take $\lambda^* = (1+\sqrt{\beta})$ as suggested in~\cite{GavishDonoho2014}. From Fig.~\ref{fig:AsymPerf}, it can be seen that for comparatively lower ranks, where the asymptotic conditions hold, the performance of the SVLET and the optimal shrinkage coincide, as expected from theorem~\ref{theorem-asymp_optim}. As the rank increases the SVLET dominates over the RMT-based shrinkage estimators, at all SNR levels. At low SNR, the SVLET performs better than the SVHT and the SVST, even at very low ranks. This is because, the SVLET does not make any assumption about the rank of underlying signal matrix. Hence, the overall improved performance by the SVLET can be attributed to its better adaptability to the singular value structure of underlying signal matrix and corrupting noise.
    \begin{figure}
    \includegraphics[width=\linewidth]{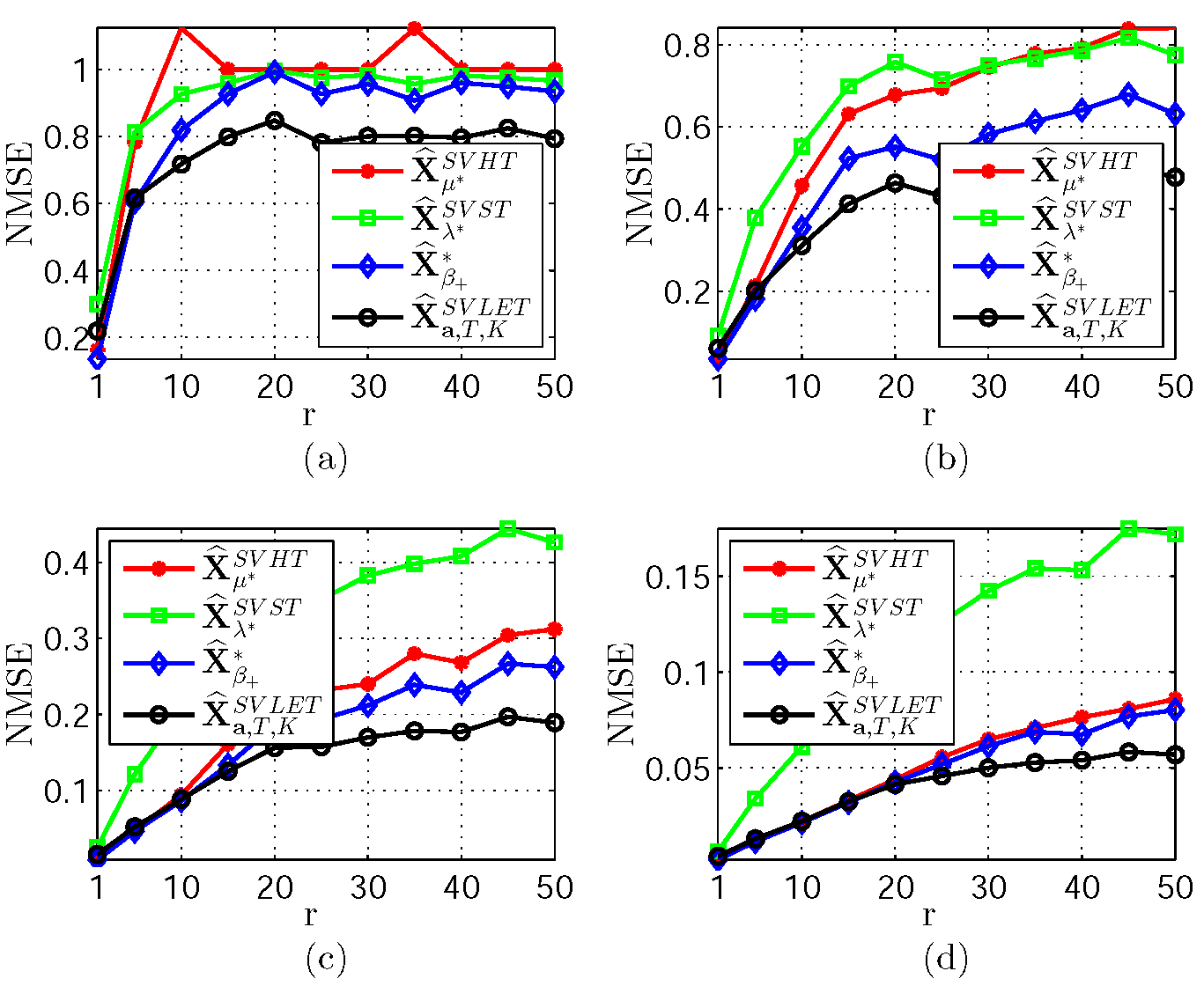}
    \caption{Comparison of average denoising performance of the proposed SVLET estimator with the asymptotic estimators for different ranks of $50\times 50$ matrix corrupted by AWGN at (a) SNR $=0.5$, (b) SNR $=1.0$, (c) SNR $=2.0$, and (d) SNR $=4.0$.}
    \label{fig:AsymPerf}
    \end{figure}
    \subsubsection{Comparison with SURE based Methods} Fig.~\ref{fig:SUREPerf} compares the SVLET with SURE based state-of-the-art shrinkage estimators for SNR = $0.5$, $1.0$, $1.5$, and $2.0$. For the SVST, the optimal soft-threshold value $\lambda$ is searched at $100$ equally spaced grids to solve (\ref{eqn:lambda-tune}). For the ATN and the SVLT, the optimal values of parameters are obtained by minimizing SURE, in the range given in Table~\ref{table:timecomp}. From Fig.~\ref{fig:SUREPerf}(a), it is clear that the SVLET outperforms all the methods, especially when the rank of underlying signal matrix is large and the noise level is high. Figs.~\ref{fig:SUREPerf}(b) and \ref{fig:SUREPerf}(d) shows the empirical optimality of the SVLET in low noise cases.


     All the methods are implemented in Matlab 8.3 (R2014a), running on Intel-i3 processor with Window 8.1 (4GB RAM) machine. Although CPU time is not a standard measure of computational cost, we use it to roughly compare the execution-time of the SVLET with the contrast methods. In this setup, the average execution-time for denoising a $50\times 50$ matrix is compared in Table~\ref{table:timecomp} for different methods. Note that, as the number of parameters increases execution-time also increases. The SVLET is the fastest one as it solves the system of linear equation for determining the optimal value of data dependent parameters.

    \begin{figure}
    \includegraphics[width=\linewidth]{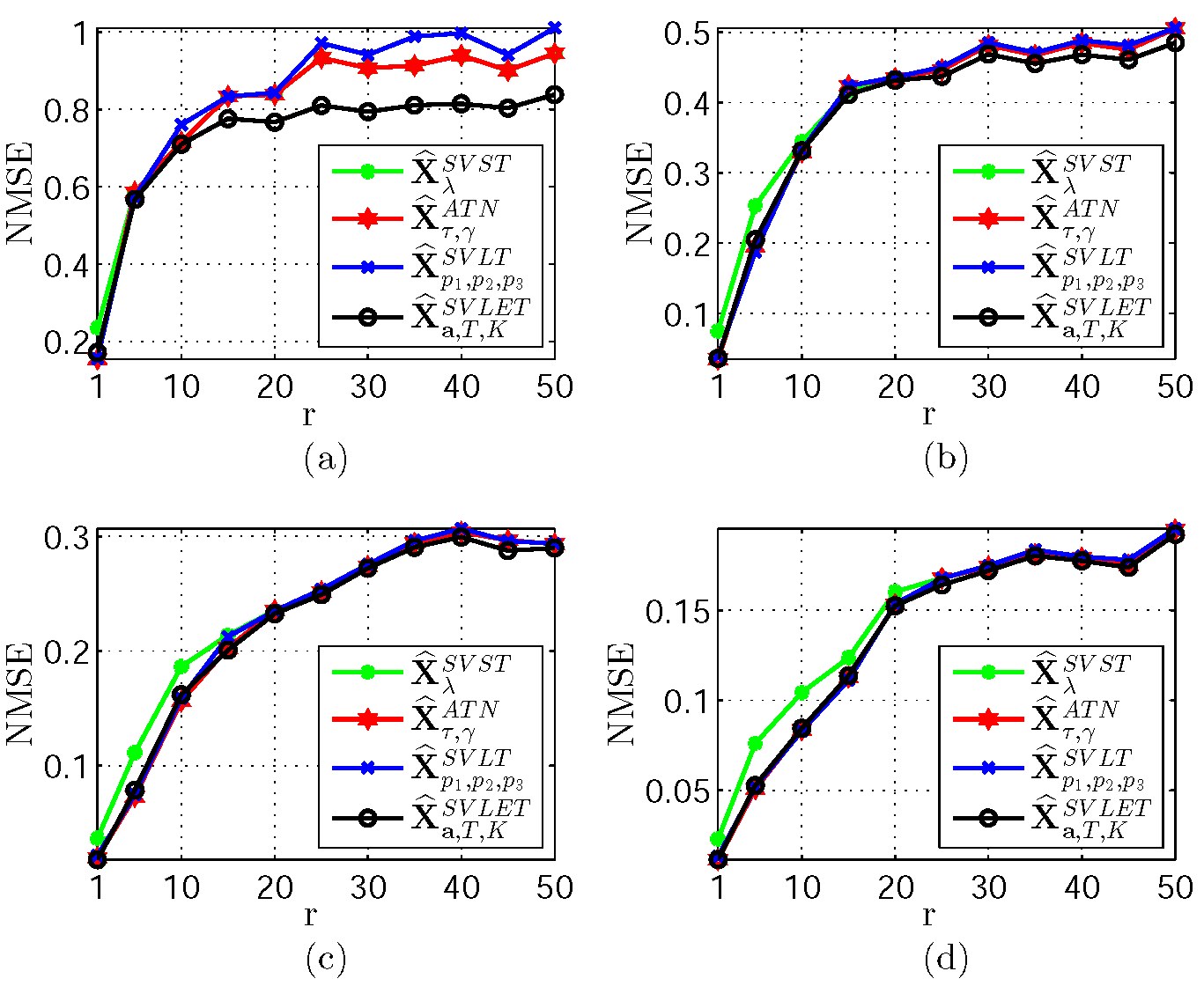}
    \caption{Comparison of average denoising performance of the proposed SVLET estimator with SURE based contrast estimators for different ranks of $50\times 50$ matrix corrupted by AWGN at (a) SNR $=0.5$, (b) SNR $=1.0$, (c) SNR $=1.5$, and (d) SNR $=2.0$.}
    \label{fig:SUREPerf}
    \end{figure}
    \begin{table}[ht]
    \begin{center}
    \caption{Comparison of the  execution time of different methods for estimating a $50\times 50$ matrix.}
    \begin{tabular}{|l||l|c|}
     \hline
     Methods & Parameter (Value/Range) & Time (in sec.) \\
     \hline \hline
     SVST &  $\lambda\in(0,0.5y_1)$ & 0.10\\ \hline
     \multirow{2}{*}{ATN}   &  $\lambda\in(0,0.5y_1)$,& \multirow{2}{*}{0.35}\\
                            & $\gamma\in[1,20]$ & \\ \hline
     \multirow{3}{*}{SVLT}  & $p_1=100$,    & \multirow{3}{*}{0.50}\\
                            & $p_2\in[1,n]$,&  \\
                            & $p_3\in(0,0.5y_1)$ & \\ \hline
     SVLET                  & $\mathbf{a}_{opt}$, $C=10$, $K=2$ & 0.04\\
     \hline
    \end{tabular}
    \label{table:timecomp}
    \end{center}
    \end{table}

	
    \section{Conclusion}
    \label{sec:Conclu}
    In this paper, we proposed a matrix estimation method by adaptive shrinkage of its singular values and shown its efficacy in recovering all rank matrices from their noisy observation. Unlike the existing methods, which select the optimal shrinkage parameters by exhaustive search, the proposed method finds one shot solution for the optimal parameters, thereby reducing the computational complexity drastically. The main contribution of the current work is in the proposition of a fast and adaptive shrinkage of the singular values based on the linear expansion of conventional singular value thresholding operator. We have theoretically shown that the proposed method is optimal in asymptotic framework as it achieves the true mean-square error and results in asymptotically optimal estimators obtained by the random matrix theory. Its denoising performance is as good as the asymptotically optimal estimators when rank of underlying signal matrix is very low, whereas it performs significantly better than the state-of-the-art methods if the rank of underlying signal matrix and the level of noise both are high. Applications of the proposed shrinkage estimator in real-time data such as in denoising of natural image and MRI sequence have also been investigated and the proposed SVLET is found to be better that the contrast methods. Exploration of SVLET in other applications and in the case of non-Gaussian noise are the courses of future research.


	\bibliographystyle{IEEEtran}
	\bibliography{referencefile}
	\end{document}